\documentclass[journal]{IEEEtran}
\ifCLASSINFOpdf
\else
\fi

\usepackage{color}
\usepackage{xcolor, soul}
\sethlcolor{yellow}
\usepackage{float}
\usepackage{makeidx}  
\usepackage{graphicx} 
\usepackage{mathrsfs}
\usepackage{arydshln}
\usepackage[tbtags]{amsmath}         
\usepackage{amssymb}                 
\usepackage{epsfig}                  
\usepackage{latexsym}                
\usepackage{array}                  
\usepackage{ifthen}
\usepackage{epstopdf}
\usepackage{booktabs}
\usepackage{subcaption}
\usepackage{theorem}
\usepackage{color}
\usepackage{comment}
\usepackage{CJK}
\usepackage{epstopdf}
\usepackage{cases,multicol}
\usepackage{textcomp}
\usepackage{tikz}

\newtheorem{theorem}{Theorem}
\newtheorem{proof}{Proof}
\newtheorem{lemma}{Lemma}

\newtheorem{defn}{Definition}

\newtheorem{fact}{Fact}
\newtheorem{remark}{Remark}

\newcommand{\be}{\begin{equation}}
\newcommand{\ee}{\end{equation}}

\newcommand{\R}{\mbox{$\mathbb{R}$}}

\begin{document}

\title{A Fresh Look on Network Synchronization} 
%
%

\author{Jilie~Zhang 
\thanks{Jilie Zhang was the School of Information Science and
Technology, Southwest Jiaotong University, Chengdu, Sichuan, 610031, P.R. China. (E-mail: jilie0226@163.com)}
%
}

%
%

\markboth{Journal of \LaTeX\ Class Files,~Vol.~14, No.~8, August~2025}%
{Shell \MakeLowercase{\textit{et al.}}: Bare Demo of IEEEtran.cls for IEEE Journals}
%



\maketitle

\begin{abstract}
This paper gives a fresh look at network synchronization. 
Here we no longer analyze it from the view of mathematics, such as graph theory, while we probe into one from control theory.
First, we analyze the synchronization region using the inner coupling matrix, giving up the routine method of studying the network structure.
The motivation comes from the inner coupling matrix that is not subject to any restrictions like network structure, such as distance and communication strength among nodes.
{\color{red}It can be configured at will to meet the synchronization performance if only the states of the local dynamic are measurable or observable and the communication topology is connected.
Thus, it is very useful for future practical engineering design.} In addition, we have an amazing finding that the network synchronization and multi-agent system consensus problems are equivalence essentially. Afterwards a unified viewpoint, that is, the essence of multi-agent consensus control is the same as that of network synchronization, is present.
Here, the equivalence relation is clearly proven and proposed. Therefore, {\color{red}we can synthesize the inner coupling matrix for network systems or the controller gain for multi-agent systems for each other.} 
Finally, {\color{red}we also present a kind of method for addressing the nonlinear complex network system.} Then the effectiveness of method is verified by taking the network of the three-
oscillator universal probe as an example.

\end{abstract}
\begin{IEEEkeywords}
Multiagent consensus, Network synchronization, Network syn-
chronizatio, Inner coupling matrix.
\end{IEEEkeywords}

%
\IEEEpeerreviewmaketitle

\section{Introduction}
The synchronization issue is an important branch for complex network research,
which is a topic that has been discussed over and over again.
Just as many literature have mentioned, there exist many phenomena of network synchronization in nature,
such as spatio-temporal chaos \cite{ZHELEZNYAK1994}, after the wonderful performance,
the audience spontaneously burst into applause \cite{Ravasz2000} and the fireflies rhythmically flash at the same frequency \cite{Mirollo1990}, ect.
However, what intrigues us is what ``mysterious force'' causes them to exhibit a strong tendency toward synchronization?
In the words of Kenneth S. Fink in \cite{Fink2000}, ``A central dynamical
question is, when is such synchronous behavior stable,
especially in regard to coupling strengths and connectivity in
the network?''

The initial understanding of collective behavior attributes to the structural properties of a network in \cite{Strogatz2001}.
While \cite{Barahona2002} is to link explicitly the Small-Word addition of random shortcuts to the synchronization of networks of coupled dynamical systems.
Afterwards, it makes the linear stability of the synchronous state be linked to an algebraic condition of the Laplacian matrix of the network by the connectivity concept.
Since the eigenvalues depend on the properties of the network structure, 
 The influences of the connection matrix eigenvalues on system synchronization are discussed in \cite{Pecora1998} by the master stability function (MSF). For example, 
the coupled system is described as one consisting of $N$ node (isolated dynamics), where the autonomous dynamics of the $i_{th}$ node is
\begin{align}\label{ee1}
\dot{{x}}_i = {F}({x}_i) + \sigma \sum_j{a}_{ij}\otimes {H}({x_j}),\quad i=1,\cdots,N 
\end{align}
where $F(\cdot)$ is the system function and $H(\cdot)$ is the inner coupling function. $l_{ij}$ are the coupling coefficients of the network structure. $\sigma$ is the coupling strength. The collective system is written as follows
\begin{align}\label{ee2}
\dot{\mathbf{x}} = \mathbf{F}(\mathbf{x}) + \sigma \mathcal{L}\otimes \mathbf{H}(\mathbf{x}),
\end{align}
where $\mathbf{x} = (x_1, x_2, \dots, x_N)$ denotes the behavior of the system. $\mathbf{F}(\mathbf{x}) = [F(x_1), F(x_2), \dots, F(x_N)]$; $\mathbf{H}(\mathbf{x}) = [H(x_1), H(x_2), \dots, H(x_N)]$. $\mathcal{L}$ is the Laplace matrix with elements $\{l_{ij}\}$. Assuming that the network is connected, then there is exactly one zero eigenvalue for $\mathcal{L}$.  $\otimes$ is the direct product.

In recent years, most of the works try to explain a central problem: “When is such synchronous behavior stable, especially
with regard to coupling strengths in the network?"\cite{Pecora1998}.
At present, it is agreed that the coupled (network) structure plays an important role for the network synchronization. 
First, by the variational of (\ref{ee2}), and then diagonalizing $\mathcal{L}$, the system (\ref{ee2}) is recast as follows:
\begin{align}\label{ee3}
\dot{\xi}_k = [D\mathbf{F} + \sigma \lambda_k D\mathbf{H}]{\xi}_k,\quad k=1,2,\cdots,N. 
\end{align}
where ${\xi}_k$ are the transverse directions of the synchronization manifold. $\lambda_k$ are the eigenvalues of $\mathcal{L}$. Based on it, the concept of master stability functions (MSF) is presented in \cite{Pecora1998}. It first links it to the numerical synchronization region of coupled dynamical systems. The result is then verified in \cite{Fink2000}, taking three coupled oscillators array as an example.

Afterward, three main criteria for the achieving dynamic synchronization of the network (\ref{ee1}) are present \cite{WANG2002,Wang2002b,Barahona2002,Stefa2007,Chen2008}, in  light of MSF.  
The first criterion depends on the smallest nonzero eigenvalue $\lambda_2$ of $\mathcal{L}$ \cite{WANG2002,Wang2002b}, that is,
\[
\lambda_2 \in [\alpha_0, \infty],\quad \alpha_0>0,
\]
where $\alpha_0$ is determined by the network parameters, involving $F(\cdot)$, $H(\cdot)$ and $\sigma$. 
Afterwards, another criterion has been proposed. It is shown in terms of the ratio of the smallest nonzero eigenvalue and the largest eigenvalue of $\mathcal{L}$ \cite{Barahona2002}, as
follows:
\[
\lambda_2/\lambda_N \in [\alpha_1, \alpha_2],\quad \alpha_1,\alpha_2>0
\]
where $\alpha_1$ and $\alpha_2$ are determined likewise by the network 
parameters. Finally, the fact that the synchronization region of the network is a union of several intervals was found in \cite{Stefa2007} and \cite{Chen2008}, since the ratio may well be located within these several intervals. These are spelled out in detail in \cite{Chen2022}.

Whereafter, \cite{Li2010} introduces MSF into the mutiagent field and analyzes the consensus region using the MSF method.
Later, \cite{Zhang2011} develops the result with LQR controller, that is, the coupling gain needs to satisfy the following consensus region,
$$
\sigma \geq \frac{1}{2 \min_{i \in \mathcal{N}} \operatorname{Re}(\lambda_i)},
$$
where \(\lambda_i\) denotes the eigenvalues of \((\mathcal{L} + G)\).
For network synchronization, in addition to the study of structural properties,
in recent years, the controllability of complex networks is also present by the literature \cite{Liu2011,Yuan2011}.
They develop analytical tools to study the
controllability of an arbitrary complex directed network, identifying the set of driver nodes with time-dependent
control that can steer the entire system to achieve synchronization.
Based on these works, there are many emerging works on pinning control, such as \cite{Sorrentino2007,Zhou2008,Yu2009,Lu2010}.
In particular, \cite{Lu2010} points out that the nodes that can access all other nodes in the digraph are referred to the pinned nodes.
However, they actually still change the structural properties of the network using the external pinning controller, that is, adding the edges between the pinning nodes and the reference node.

Therefore, analyzing network structure has become the mainstream approach for revealing the phenomenon of synchronization, since it is intuitive and easy to analyze using the existing mathematical tools, such as graph theory. 
However, there is an unavoidable issue, that is, how to address it if the network structure is unalterable.
For instance, it is unable to communicate between the pinning nodes and the reference node
due to the limitations of distance and signal strength. Thus, the pinning control is not available. 
In addition, for spontaneous synchronization of the network,
some communication properties are also unchangeable.
For instance, due to physical isolation or because the signal strength is bounded among nodes,
they are unable to interact with each other. Then, all analytical methods based on the structural properties of the network fail.
Namely, how to respond to Kenneth S. Fink's problem without using the point view of the network structure.
Here, we try to explore this problem from a fresh look.

{\color{red}In practice, at present most of the works pay overly attention to the impact of network structure on synchronization,
while neglecting another important factor, that is, the inner coupling matrix.
In fact, it is obvious to see that the inner coupling matrix plays an important role in the synchronization of the network, from MSF in \cite{Pecora1998}, as follows
$$
\left( \mathrm{DF} + \sigma \lambda_k \mathrm{DH} \right)<0.
$$
Therefore, when the network structure is fixed, the proper inner coupling matrix can also enable the system to synchronize.
The greatest advantage lies in that it is not restricted by the network structure and can be configured at will to meet the synchronization performance if only the states of the local dynamic are measurable and observable. Thus, it is also easy to implement for future practical engineering design.}
Here, the implementation method is first introduced systematically taking the example of linear system network synchronization in Section \ref{s1}.
In addition, the method can improve process (such as synchronization time and overshot) of synchronization effectively, even it can achieve the ability to place the poles. This paper unveils the ``mysterious force'' by analyzing the impact of the inner coupling matrix on network synchronization.
Next, we also present a unified viewpoint for network synchronization and multiagent consensus, in Section \ref{s2}. 
We clearly setup the relation between them and show the equivalence of both questions: configuration problem for the inner coupling matrix and the design problem on control gain.
Finally, we also present a kind of method to address the nonlinear complex network system in Section \ref{s3}. The network of the three-oscillator universal probe is used to verify the effectiveness of our method.






\section{Preliminaries}
\subsection{Description of network topology}

The topology among network nodes can be encoded by graph theory.
Define $\mathcal{T} = (\mathcal{V}, \mathcal{E}, \mathcal{A})$ to describe the network/multi-agent system with ${N}$ nodes, where $\mathcal{V}=\{{v}_1,\ldots,{v}_{{N}}\}$ is the node set.  $\mathcal{E}$ is the edge set with $\mathcal{E} \subseteq \mathcal{V} \times \mathcal{V}$, and $\mathcal{A}=[a_{ij}]$ is the coupling matrix. An edge of the topology $\mathcal{T}$ is denoted by $\nu_{ij}=(v_j,v_i)$, which graphically represented by the link with the arrow from node $j$ to node $i$. It means the information can flow from node $v_j$ to node $v_i$.
We assume that all edges are positive, then $a_{ij} > 0$ if and only if $\nu_{ij} \in \mathcal{E}$; otherwise $a_{ij} = 0$. For notational simplicity, denote $i\in \mathcal{I}=\{1,2,\ldots,{N}\}$.
According to the above description, the following definition is given.
\begin{defn}[Laplacian Matrix]\label{D1} The Laplacian matrix of a topology is defined as $ \mathcal{L}:=[l_{ij}]=\mathcal{D}-\mathcal{A}$, where $\mathcal{D} = \mbox{diag} \{l_{ii}\} \in R^{{N} \times {N}}$ is the in-degree matrix of topology $\mathcal{T}$, with $l_{ii} = \sum\limits_{{j=1, i \neq j}}^{{N}} a_{ij}$ being the in-degree of node $v_i$.
\end{defn}
\begin{remark}
It needs to be emphasized that the Laplacian matrix $\mathcal{L}$ is used to describe the relation among adjacent nodes in the multiagent field.
While in the complex network field, it is defined by the connection matrix $\mathbf{G}$ \cite{Pecora1998,Fink2000}.
There is the relation $\mathcal{L}=-\mathbf{G}$ between them. Therefore we use the symbol $\mathcal{L}$ uniformly here. \hfill$\blacksquare$
\end{remark}
Note that $\lambda_{max}(\cdot)$ indicates the maximum eigenvalue of specified matrix. $Re(\cdot)$ indicates that the real part of the complex number.

\begin{lemma}\label{l1}
Assume that a complex number $\varrho$ is in the right half-plane and
all $Ger\check{s}gorin$ discs of a complex matrix $\mathcal{Z}=[z_{ij}]$ are in the left half-plane, and then 
all eigenvalues of the matrix $\varrho \mathcal{Z}$ is still located in the open left half-plane,
if only
\begin{align}
r_{ii}\sin|\theta|&\geq\tilde{R}_i,\quad \forall i \label{E12}\\
\varphi_{ii} - 2|\theta| &>\frac{\pi}{2},\quad \forall i\label{E11}
\end{align}
with $\varphi_{ii}$ is the argument of $z_{ii}$, and $\theta$ is the argument of $\varrho$.
\end{lemma}
\begin{proof} Assume that $z_{ij}=\alpha_{ij}+i\beta_{ij}$ with $\alpha_{ij}<0$, and $\varrho=a\pm i b$ with $a>0$,
then we can rewrite them in polar form,
$$z_{ij}=r_{ij}e^{i\varphi_{ij}},\quad \mbox{and}\quad {\varrho}=\gamma e^{i\theta}$$
where $r_{ij}=|z_{ij}|=\sqrt{\alpha_{ij}^2+\beta_{ij}^2}$ and  $\varphi_{ij}= \mbox{arctan}\left(\frac{\beta_{ij}}{\alpha_{ij}}\right)$;
$\gamma=|\varrho|=\sqrt{a^2+b^2}$ and  $\theta= \mbox{arctan}\left(\frac{b}{a}\right)$.

And then, all $Ger\check{s}gorin$ discs of $\mathcal{Z}$ are  described as
$$
G_i = \{ z \in \mathbb{C} \mid |z - z_{ii}| \leq \tilde{R}_i \},\ \tilde{R}_i = \sum_{j \neq i}^n |z_{ij}| = \sum_{j \neq i}^n r_{ij}, \ \forall i,
$$
and are located in the left half-plane, such as the green disc in Fig. \ref{f22}.
While $\theta \in[-\frac{\pi}{2},\frac{\pi}{2}]$, because $\varrho$ is in the right half-plane.

Next, before analyzing the matrix $\varrho \mathcal{Z}$, we first consider $\mathcal{Z}  e^{i\theta}$, that is,
$$\mathcal{Z} e^{i\theta}:=[z_{ij}e^{i\theta}]=[ r_{ij} e^{i(\varphi_{ij}+\theta)}].$$
Therefore, $|z_{ij}e^{i\theta}|=r_{ij}$, for $i\neq j$.
Obviously, the factor $e^{i\theta}$ does not change the radius $\tilde{R}_i$ of the $Ger\check{s}gorin$ discs in the matrix $\mathcal{Z}$.
However, it rotates the vector (the center of the disc $O$) clockwise to $O_1$ by the angle $\theta$ in the worst case ($\varrho=a-ib$).
So, the constraints (\ref{E12}) and (\ref{E11}) guaranty all $G_i$ are in the left half-plane, see Fig. \ref{f22} in detail.
While the factor $\gamma$ just scales proportionally the center and radius of these discs, it does not change the conclusion before, such as disc $O_1^\prime$.
Therefore, all eigenvalues of the
matrix $\varrho \mathcal{Z}$ must be located in the open left half-plane. \hfill$\blacksquare$
 \begin{figure}[H]
  \centering
  \includegraphics[width=3.5in]{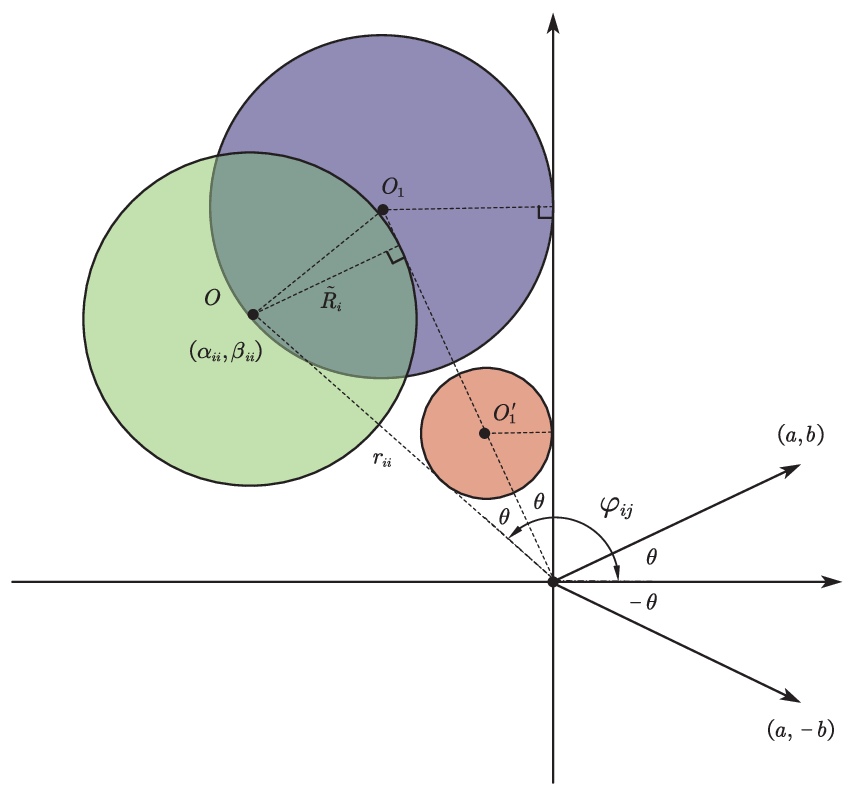}
  \centering \caption{Schematic diagram on Lemma \ref{l1}}
\end{figure}\label{f22}
\end{proof}
\begin{lemma}\label{l2}
Assume that $\mathcal{Z}$ is a complex matrix with all $Ger\check{s}gorin$ discs being in the right half-plane or left half-plane, then the matrix $Re(\mathcal{Z})$ does not change the sign of every eigenvalue.
\end{lemma}

\begin{proof}
Using the symbol definitions as stated in the proof of Lemma \ref{l1}, all the $Ger\check{s}gorin$ discs are shown by
$$
G_i = \{ z \in \mathbb{C} \mid |z - z_{ii}| \leq \tilde{R}_i \},\quad\tilde{R}_i = \sum_{j \neq i}^n r_{ij},\quad \forall i.
$$

While the ones of the matrix $Re(\mathcal{Z})$ are
$$
G_i = \{ z \in \mathbb{C} \mid |z - \alpha_{ii}| \leq \tilde{R}_i^\prime \},\quad\tilde{R}_i^\prime = \sum_{j \neq i}^n \alpha_{ij},\quad \forall i.
$$

Since $\alpha_{ij}\leq r_{ij}$, $\tilde{R}_i^\prime\leq \tilde{R}_i$. This implies that the $Ger\check{s}gorin$ discs are reduced in size. Therefore, the result is true. \hfill$\blacksquare$

Both results will be used in the design of the inner coupling matrix for the complex network in the next section.
\end{proof} 

\section{Synchronization on complex network}\label{s1}
Here we discuss the design problem of the inner coupling matrix. It refers not only to the synchronization for coupled dynamics,
but also to their transient performance, such as synchronization speed and overshot.

First we consider the linear model of the complex network,
\begin{align}\label{E2}
\dot{x}_i=&A{x_i}+ \sigma \sum_{j=1}^N{G_{ij}H x_{j}}, i=1,2,\ldots,N
\end{align}
where $x_i \in \R^{n}$ is the system state.
$A$ is the isolated (uncoupled) dynamics. $H$ is the inner coupling matrix for the output of each node; 
and then $\sigma$ is the coupling strength. $N$ is the number of the node.
$G_{ij}$ is the connecting weight associated with other nodes. 

{\color{red}\textbf{Problem:} The aim is to design an inner coupling matrix to meet the desired transient performance for the complex network (\ref{E2}), namely, the pole can be assigned at will in each transverse direction to the synchronization manifold, by synthesizing the inner coupling matrix.  }

A compact form of Eq. (\ref{E2}) is
\begin{align}\label{E33}
\dot{x}=&\hat{A}{x}+ \sigma G \otimes H x,
\end{align}
where $\hat{A}=(I_N\otimes A)$. $G$ is the coupling configuration (connection) matrix, such as 
$$
G_{ii}=\sum_{j = i, j\neq i}^N G_{ij}.
$$
It is also constant in the synchronization manifold.

For convenience in discussion later, letting $G=-\mathcal{L}$ and $\tilde{H}=-H$, Eq. (\ref{E33}) can be rewritten as 
\begin{align}\label{E3}
\dot{x}=&\hat{A}{x}+ \sigma \mathcal{L} \otimes \tilde{H} x.
\end{align}
There exists a nonsingular matrix $U \in \R^{N\times N}$, $UU^{-1}$, without loss of generality, such that
\begin{align*}
U^{-1}\mathcal{L}U=\mbox{diag}\bigl(J_{k_1}(\lambda_1),\ J_{k_2}(\lambda_2),\ \dots,\ J_{k_N}(\lambda_N)\bigr), \\ 
\end{align*}
where 
$0=\lambda_1,\lambda_2,\cdots,\lambda_{{N}}$ are the eigenvalues of the Laplace matrix and 
\begin{align*}
J_{k_i}(\lambda_i) =
\begin{bmatrix}
\lambda_i & 1 & & \\
& \lambda_i & \ddots & \\
& & \ddots & 1 \\
& & & \lambda_i
\end{bmatrix}.
\end{align*}

And then Eq. (\ref{E3}) is equivalent to  a set of decoupled systems
\begin{align}\label{E4}
\dot{\hat{x}}_i=&(A+ \sigma \lambda_i \tilde{H}) \hat{x}_{i}.
\end{align}
It is well known that the network system achieves synchronization, if the designed coupling matrix $\tilde{H}$ satisfies
\begin{align}\label{E5}
(A + c \lambda_i \tilde{H})\ \mbox{is Hurwitz},\quad i=2,3,\ldots,N.
\end{align}
By the linear transformation $\hat{x}_i=P\tilde{x}_i$ agian, without loss of generality,  Eq.(\ref{E4}) is rewritten as
\begin{align}\label{E6}
\tilde{x}_i=&(\Lambda+ \sigma {\lambda}_i \mathcal{H}) \tilde{x}_{i},
\end{align}
where  $\mathcal{H}:=[h_{ij}]=P^{-1}\tilde{H}P$, $\Lambda=P^{-1}AP=
\mbox{diag}\bigl(J_{\kappa_1}(\bar{\lambda}_1),\ J_{\kappa_2}(\bar{\lambda}_2),\ \dots,\ J_{\kappa_N}(\bar{\lambda}_n)\bigr)$. ${\bar{\lambda}}_i$ is the eigenvalue of the system matrix $A$.
{\color{red}In fact, we can not only configure the coupling matrix $\mathcal{H}$, such as
\begin{align}\label{E7}
(\Lambda+ \sigma {\lambda}_{i} \mathcal{H})\ \mbox{is Hurwitz},\quad i=2,3,\ldots,N,
\end{align}
but also assign the pole in each transverse direction to the synchronization manifold.}

Given that the topology of the network is fixed being subject to actual conditions, we first start from the simple case, that is, the topology is undirected.
Without loss of generality, let the coupling strength $\sigma=1$ in the following statement.

\subsection{Undirected topology}

The eigenvalues of the Laplace matrix $\lambda_{i}$ is a real number corresponding to an undirected topology. Then it has $\lambda_1=0\leq\lambda_2\leq\cdots\leq\lambda_{{N}}$.
Eq.(\ref{E7}) is true as long as we design the proper $\mathcal{H}$ with the minimum eigenvalue $\lambda_2$.
For simplicity, we set $\mathcal{H}$ as a diagonal matrix with the same element $h_{ii}$, sucn as
$$
\lambda_2Re(h_{ii})+ \lambda_{\max}(Re(\bar{\lambda}_k))<0, \quad k=1,\cdots,n.
$$
Then, it has
\begin{align}\label{E8}
Re(h_{ii})<- \frac{\lambda_{\max}(Re(\bar{\lambda}_k))}{\lambda_2}.
\end{align}

By Lemma \ref{l2}, the inner coupling matrix ${H}$ is
\begin{align}\label{E9}
H=-\tilde{H},
\end{align}
where $\tilde{H}= Re(P\mathcal{H}P^{-1})$.

{\color{red}It is too rough, although the above method is simple. Many evolution of the spontaneous synchronization can not be govern accurately.
Actually, the configuration of $H$ can be set using the more accurate method in all transverse directions to the synchronization manifold. 
For example, the desired pole $p_i$ can be assigned by the following way,
\begin{align}\label{E80}
Re(h_{ii})=- \frac{\lambda_{\max}(Re(\bar{\lambda}_k))-p_i}{\lambda_2}.
\end{align}}

It is worth noting that the coupling matrix $H$ must be a real matrix for implementation.
Therefore, $H$ must take the real part of $P\mathcal{H}P^{-1}$, which is calculated by the numerical calculation.
{By Lemma \ref{l2}, the operation does not change the sign of the eigenvalues of the matrix $H$. In addition, since the radii of the discs are reduced, the eigenvalues are converging towards the center of the circle. Thus, it is more beneficial for the stability of the system.}

\textbf{Example 1:}
Assume that the isolated dynamic system of the complex network system is
$$
A = \begin{pmatrix}
   0  & -10  &  0 &  0 &   0\\
   10  &  0  &  0 &  0 &   0\\
   0  &  0  &  -3 &  0 &   0\\
   0  &  0   & 0 & 0 &   -10\\
   0  &  0   & 0 &  0 &   0\\
\end{pmatrix},$$
and the Laplace matrix is symmetrical with respect to the undirected topology, such as
$$
\mathcal{L}=\begin{pmatrix}
1 & -1 & 0 & 0 & 0 \\
-1 & 2 & -1& 0 & 0 \\
0 & -1 & 2 & -1 & 0 \\
0 & 0 & -1 & 2 & -1 \\
0 & 0 & 0 & -1 & 1
\end{pmatrix},$$
with $\lambda_2=0.3820$.

By the linear transformation, we can get
$
\Lambda =\mbox{diag}(\pm 10i,\  -3,\  0,\ 0)$. 
In light of (\ref{E8}), we take $h_{ii}$ as $20$ for all $i$ in $\mathcal{H}$ and obtain the coupling matrix $H$ by Eq.(\ref{E9}).
$$
{H}_1= \begin{pmatrix}
-20 & 0 & 0 & 0 & 0 \\
0 & -20 & 0 & 0 & 0 \\
0 & 0 & -20 &0 & 0 \\
0 & 0 & 0 & -20 &0 \\
0 & 0 & 0 & 0 & -20
\end{pmatrix}.
$$
From the synchronization evolution of the network in Fig.\ref{subfig:1},
we clearly observe that the states of the five nodes tend to synchronize, respectively, in each mode (The red and green lines are $x_{i1}$ and $x_{i2}$, respectively. They are persistent oscillations, since both modes are imaginary numbers; the blue $x_{i3}$ is exponential decay; the black $x_{i4}$ is a ramp response; while the yellow $x_{i5}$) remains a constant, ultimately.

To show the influence of the inner coupling matrix $H$ on the evolution of synchronization, we redesign $h_{44}$ and $h_{55}$ as $30$ without other change.
By Eq.(\ref{E9}), $H$ is calculated as
$$
{H}_2= \begin{pmatrix}
-20 & 0 & 0 & 0 & 0 \\
0 & -20 & 0 & 0 & 0 \\
0 & 0 & -20 &0 & 0 \\
0 & 0 & 0 & -30 &0 \\
0 & 0 & 0 & 0 & -30
\end{pmatrix}.
$$
It results in a synchronization evolution in Fig.\ref{subfig:2}. 
Apparently, the synchronization evolution for the states $x_4$ and $x_5$ of every node 
(that is, the black and yellow lines $x_{i4},x_{i5}$) is much faster than that in Fig.\ref{subfig:1}.
The reason lies in that the real parts of their Lyapunov exponents are much smaller than before in their transverse directions, further shortening their synchronization time.
By the simulations in Fig.\ref{subfig:1} and Fig.\ref{subfig:2}, a proper design of the $H$ matrix not only enables the system to achieve synchronization, but also improves its evolution performance.

Next, we introduce the coupling matrix $\mathcal{H}_3$ between $x_{i4}$ and $x_{i5}$ by setting $h_{45,}=h_{5,4}=10$ in $\mathcal{H}_2$. Then, $H$ matrix is calculated as 
$$
{H}_3= \begin{pmatrix}
-20 & 0 & 0 & 0 & 0 \\
0 & -20 & 0 & 0 & 0 \\
0 & 0 & -20 &0 & 0 \\
0 & 0 & 0 & -30 & 0 \\
0 & 0 & 0 & 0 & -10
\end{pmatrix}.
$$

Actually, the result is counterintuitive, that is, by the coupling matrix ${H}_3$, enhancing the coupling between states did not improve performance, in contrast, even worse, such as the black and yellow lines converge more slowly, from Fig.\ref{subfig:3}. 
The reason is that the two added coupling gains $h_{45,}$ and $h_{5,4}$ increase the eigenvalues in the corresponding to the transverse direction.

\begin{remark}
From Example 1, we know that the coupled matrix ($\mathcal{H}_3$), which seems to be complex, does not necessarily lead to better results than the simple situation ($\mathcal{H}_2$).
In addition, here the desired performance can be completed by pole placement.
{Therefore, how to design the $H$ to meet the desired performance with the least coupling
is a new challenging work.} Of cause, the diagonal matrices should be considered first. \hfill$\blacksquare$
\end{remark}

\subsection{Directed topology}
Now, Let us discuss the more general case, i.e., directed topology.
In this case, the eigenvalues of the Laplace matrix $\lambda_{i}$ may well be a complex number.
If there is one (or multiple) pair of complex eigenvalues to ensure that Eq.(\ref{E7}) holds,
the matrix $\mathcal{H}$, which consists of the complex matrix, can steer the network to achieve synchronization by  Eq.(\ref{E9}).
Also, for the sake of simplification, we show the case, that is, one pair of complex eigenvalues.

First, take $\lambda_2$ as the eigenvalue of the largest argument.
Assume that $\lambda_2=\alpha_2 -i\beta_2$ (the worst case), then its amplitude is $|\lambda_2|=\sqrt{\alpha_2^2+\beta_2^2}$,
and the argument is $-\theta_{max}$. The pair of complex vectors $h_{ii,i+1i+1}=\alpha_{ii} \pm i\beta_{ii}$ can be designed in diagonal matrix $\mathcal{H}$, such that
\begin{align}\label{E10}
\quad\varphi_{ii} - \theta_{max}&>\frac{\pi}{2}
\end{align}
where $\varphi_{ii}=arctan{\frac{\beta_{ii}}{\alpha_{ii}}}$, see Fig. \ref{f2}. By Lemma \ref{l1}, we know that $\lambda_2 \mathcal{H}$ is negative definite with the constraint (\ref{E10}).
\begin{figure}[H]

    \begin{subfigure}{3.5in}
        \includegraphics[width=3.0in]{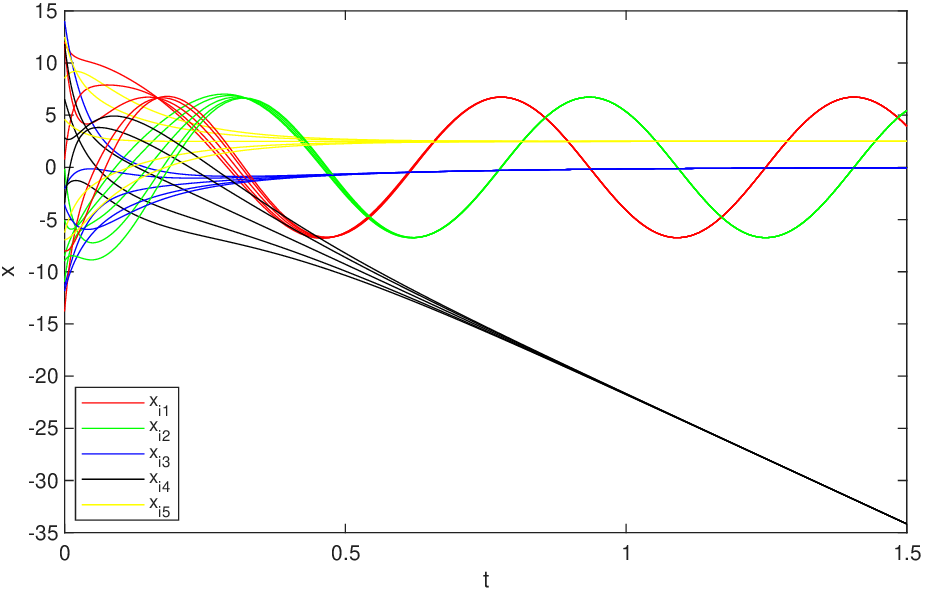}  
        \caption{The synchronization evolution for ${H}_1$}
        \label{subfig:1}
    \end{subfigure}
    \vspace{0.5cm}  

    \begin{subfigure}{3.5in}
        \includegraphics[width=3.0in]{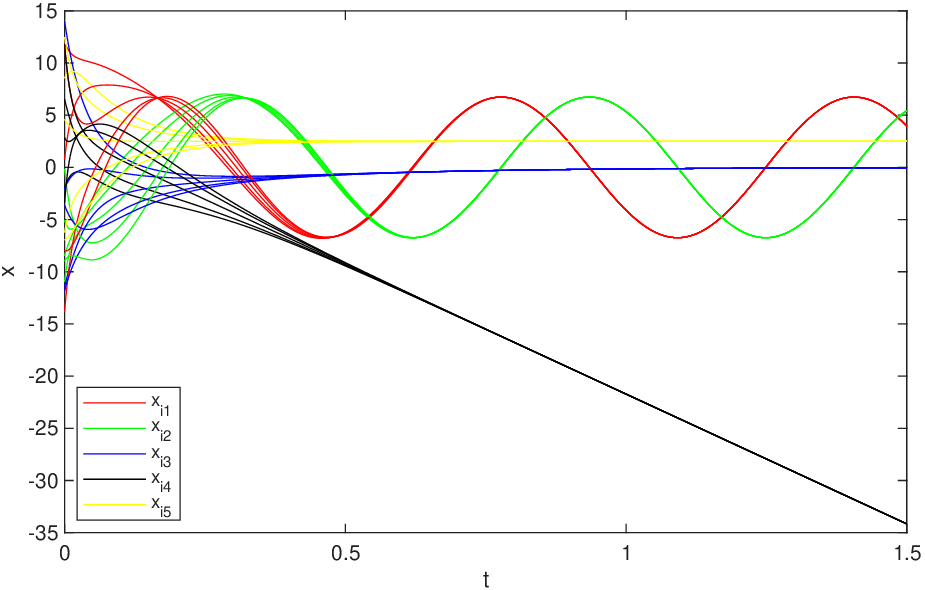}  
        \caption{The synchronization evolution for ${H}_2$}
        \label{subfig:2}
    \end{subfigure}
        \vspace{0.5cm}

      \begin{subfigure}{3.5in}
        \includegraphics[width=3.0in]{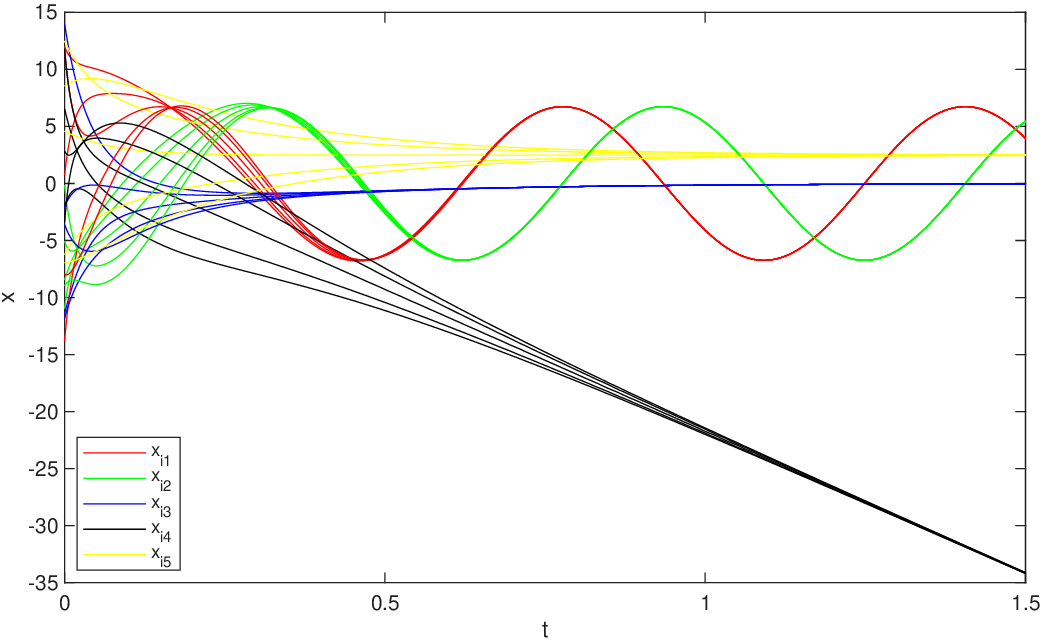}  
        \caption{The synchronization evolution for ${H}_3$}
        \label{subfig:3}
    \end{subfigure}
\centering
    \caption{Configuration results for the different inner coupling matrix $H$ over undirected topology}
    \label{f11}
\end{figure}

Secondly, Eq.(\ref{E7}) holds, if the selective $\mathcal{H}=[h_{ij}]$ satisfies the following condition,
\begin{align}\label{E100}
\lambda_2Re(h_{ii})+ \lambda_{\max}(Re(\bar{\lambda}_k))<0.
\end{align}
{\color{red}In addition, the desired pole $p_i$ can be assigned by the following way,
\begin{align}\label{E800}
Re(h_{ii})=- \frac{\lambda_{\max}(Re(\bar{\lambda}_k))-p_i}{\lambda_2}.
\end{align}}
 \begin{figure}[H]
  \centering
  \includegraphics[width=3in]{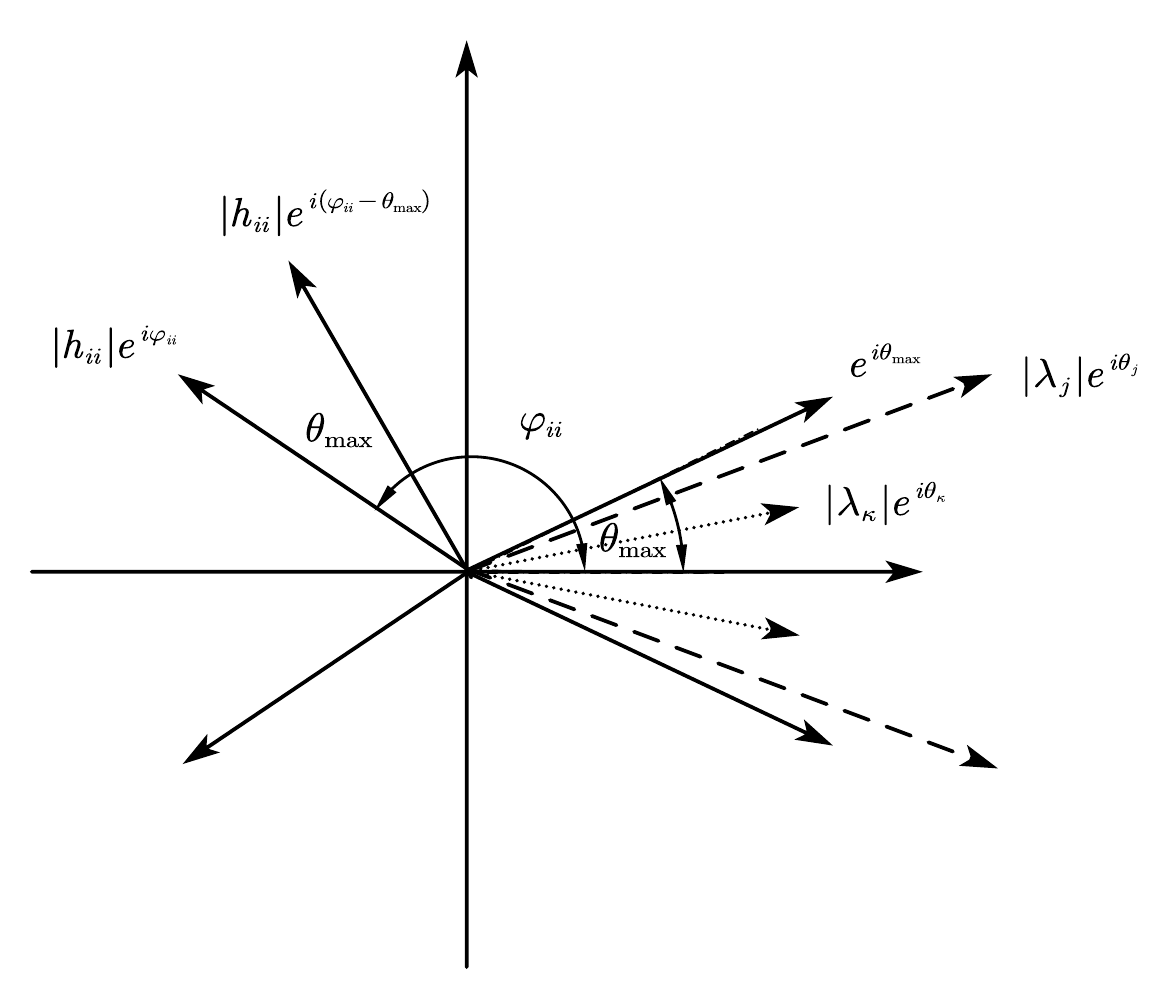}
  \centering \caption{Schematic diagram of the matrix $\mathcal{H}$ over the directed topology}\label{f2}
\end{figure}
\begin{remark}
From Fig. \ref{f2}, if only Eq. (\ref{E10}) holds, the eigenvalues of $\lambda_2 \mathcal{H}$ are negative, no matter how large the amplitude of $\lambda_i$ is.
However, it only ensures the network synchronization with the constraint (\ref{E100}). While the improvement in performance depends on $\varphi_{ii}$.
The proper argument $\varphi_{ii}$ is more useful for improving the overshot corresponding to the transverse direction, since $\varphi_{ii}$ can change the argument of the poles of the system.
Therefore, this kind of configuration method not only ensures network synchronization but also improves the transient performance of synchronization.
In addition, Eq.(\ref{E7}) also holds by choosing the enough large coupling strength $\sigma$.
Namely, whether the coupling strength $\sigma$ or matrix $\mathcal{H}$ both must ensure the Lyapunov exponents is negative correspond to transverse direction.
Finally, the inner coupling matrix $H$ can be designed by Eq.(\ref{E9}). The differences lie in that $\sigma$ adjusts the stability and performance of the overall system,
while $\mathcal{H}$ can precisely improve the synchronization performance of the transverse direction. \hfill$\blacksquare$
\end{remark}
Note that, in fact, by Lemma \ref{l1} matrix $\mathcal{H}$ can be configured in a more diverse manner in the form of a non-diagonal matrix.

\textbf{Example 2:}
Here, assume the dynamic of system is
$$
A = \begin{pmatrix}
   -1  & -3\\
   3  &  1
\end{pmatrix},$$
and the topology is directed with complex eigenvalue $\lambda_{3,4}= 0.9924\pm 0.5131i$ ($\theta_{max}=27.3399^\circ$), such as
$$
\mathcal{L}=\begin{pmatrix}
1 &  0 & 0 & 0 & -1 \\
-1 & 2 & -1& 0 & 0 \\
0 & -1 & 2 & -1 & 0 \\
0 & 0 & -1 & 2 & -1 \\
0 & 0 & 0 & -1 & 1
\end{pmatrix}.$$
In light of (\ref{E10}), we take $h_{1,2}=1\pm 0.5i$ with the argument $\varphi_{ii}=\pm 153.4349^\circ$ and get the inner coupling matrix $H$ by Eq.(\ref{E9}),
$$
{H}_4=\begin{pmatrix}
-1.1768 & -0.5303  \\
0.5303 & -0.8232 \\
\end{pmatrix}.$$

From Fig.\ref{subfig:4}, the two states of five nodes achieve agreement, respectively. The red lines denote the responses for the first mode of dynamic $A$. The green lines are those of the second one.

In order to show that the argument of the designed $h_{ii}$ in $\mathcal{H}$ impacts the system performance,
we reserve a greater margin for the argument ($\varphi_{ii}=\pm 174.2894^\circ$), taking $h_{1,2}=1\pm 0.1i$.
The coupling matrix $H$ can be calculated as
$$
{H}_5=\begin{pmatrix}
-1.0345 & 0.1061  \\
-0.1061 & -0.9646 \\
\end{pmatrix}.$$

Apparently, we observe that the argument margin is greater, the overshot of synchronization error is smaller, from the Figs. \ref{subfig:4} and \ref{subfig:5}. Therefore, the overshot of the synchronization process is clearly improved.


\section{A Unified Viewpoint: Multia-gent Systems and Synchronization of Complex Networks}\label{s2}
In terms of the relationship between the consensus of multi-agent systems and the synchronization of complex networks, the unified viewpoint is presented by the literature \cite{Li2010}.
However, it just unifies the stability domain of both systems,  reaching agreement by the master stability functions.
Here we elaborate the equivalence property of both intrinsically by the following theorem.
\begin{theorem} \label{p1}
The synchronization problem of the complex network (\ref{E2}) is intrinsically the same as the consensus problem of the multi-agent system (\ref{e11}) and (\ref{e12}),
by the following relational expression,
\begin{align} \label{e14}
H=-BK,
\end{align}
where $H$ is the inner coupling matrix of the complex network system. $B$ is the input matrix of the multi-agent system.
$K$ is the designed feedback gain of the controller.
\end{theorem}
Here we recall the multi-agent systems,
\begin{align} \label{e11}
\dot{x}_i (t) = Ax_i (t) + Bu_i (t),
\end{align}
where $x$ is the system states of agent $i$; $u_i$ is the control input of the agent $i$.
$A$ is the system matrix; $B$ is the control matrix. $(A,B)$ is controllable.

The control protocol is
\begin{align} \label{e12}
{u}_i (t) = cK \left( \sum_{j=1}^N a_{ij}(x_i-x_j) \right),
\end{align}
where $K$ is the feedback gain of the controller. $c$ is the coupling strength.

\begin{figure}[H]

    \begin{subfigure}{3.5in}
        \includegraphics[width=3.5in]{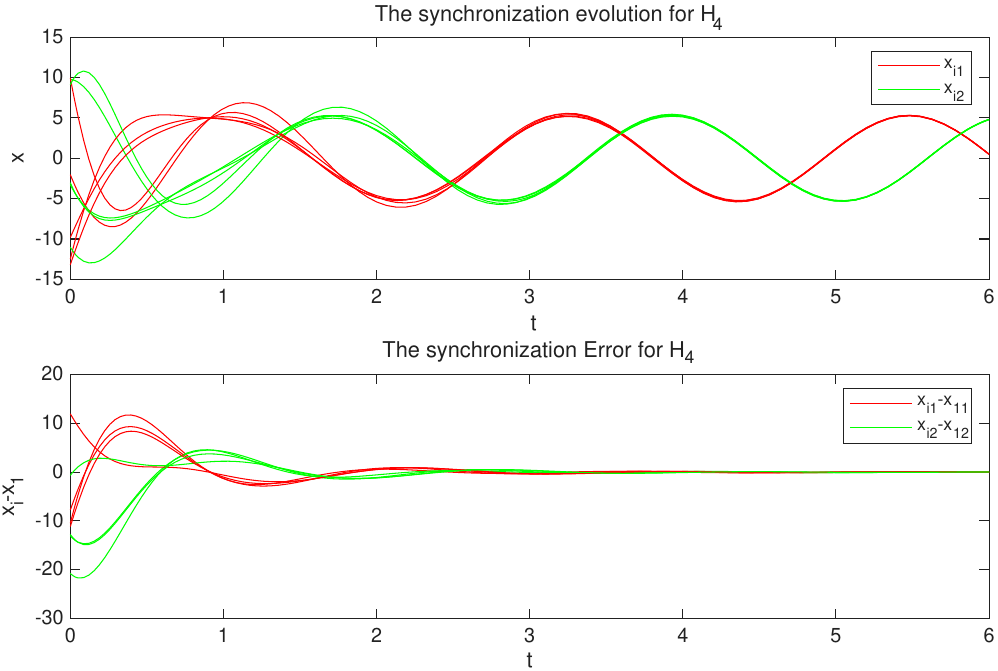}  
        \caption{The synchronization evolution for ${H}_4$}
        \label{subfig:4}
    \end{subfigure}
    \vspace{0.5cm}  

    \begin{subfigure}{3.5in}
        \includegraphics[width=3.5in]{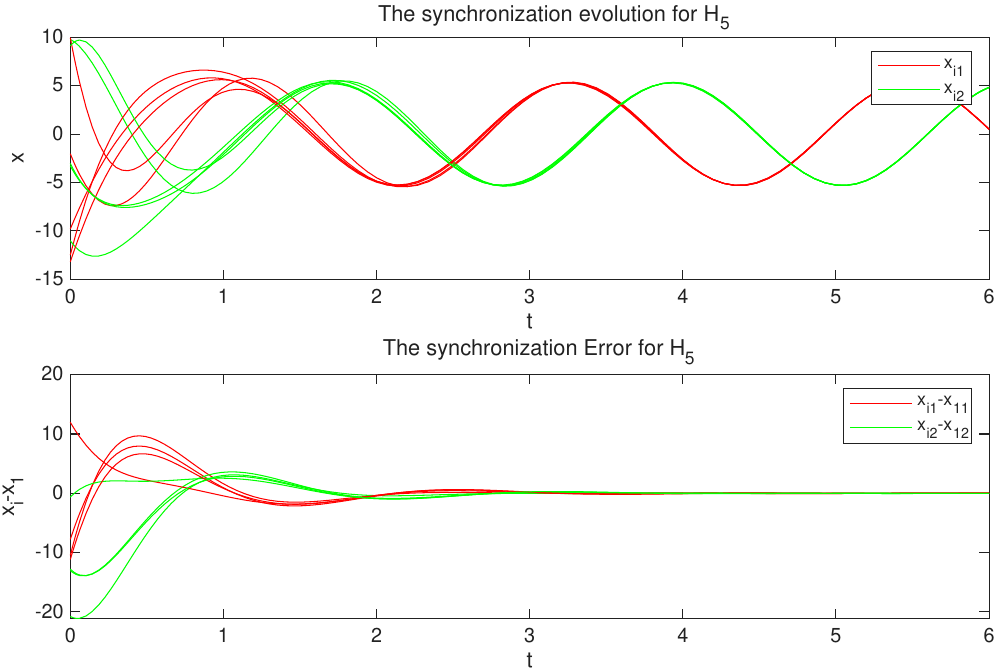}  
        \caption{The synchronization evolution for ${H}_5$}
        \label{subfig:5}
    \end{subfigure}
\centering
    \caption{Configuration results for the different coupling matrix $H$ over directed topology}
    \label{f3}
\end{figure}
\begin{proof}

Inserting (\ref{e12}) into (\ref{e11}),
\begin{align} \label{e13}
\dot{x}_i (t) &= Ax_i (t) +c BK \left( \sum_{j=1}^N a_{ij}(x_i-x_j) \right), \nonumber\\
              &= Ax_i (t) +c BK\left( \sum_{j=1}^N {a}_{ij} x_{j} \right), \nonumber\\
              &= Ax_i (t) +c \sum_{j=1}^N {a}_{ij} BK x_{j},
\end{align}
where
$$
{a}_{ii}=\sum_{j = i, j\neq i}^N a_{ij}.
$$
Compared with the complex network model (\ref{E2}),
$${H}=-BK.$$
Thus, Theorem \ref{p1} is true. \hfill$\blacksquare$
\end{proof}
\begin{remark}
From Theorem \ref{p1}, we know that the complex network synchronization problem can be solved by designing the controller for the multi-agent system, and vise versa. \hfill$\blacksquare$
\end{remark}
\subsection{Configure the inner Coupling matrix $H$ by the controller gain of multi-agent system}
Here we propose a new kind of method to configure the inner coupling matrix $H$ with the aid of the controller design for multi-agent system

First, the controllable matrix $B$ must be constructed according to the system $A$, so that $(A,B)$ is controllable.
Then design the control gain $K$ by the consensus method for the multi-agent system, such as LQR, Lyapunov function method, etc.
Finally, calculate the inner coupling matrix of the complex networks system by Eq. (\ref{e14}).

\textbf{Example 3:}
Here, to verify the effectiveness of our method, we use the agent dynamics in \cite{Li2011} as the isolated dynamic of the complex networks
$$
A = \begin{pmatrix}
   -2  & 1.5 \\
   -1  &  1
\end{pmatrix},
$$
and the Laplacian matrix is associated with the communication topology in Fig. \ref{f55}. 
$$
  L =\begin{pmatrix}
       5&  -1&  -1&  -1&  -1&   0\\
      -1&   2&  -1&   0&   0&   0\\
      -1&  -1&   2&   0&   0&   0\\
      -1&   0&   0&   2&  -1&   0\\
      -1&   0&   0&  -1&   3&  -1\\
       0&   0&   0&   0&  -1&   1
       \end{pmatrix}.
$$

Next, we configure the inner coupling matrix $H$ for this network using the multi-agent system method.
First, we construct a controllable matrix $B=[1\quad -1]^\prime$ (in fact, it is also from \cite{Li2011}), according to the system $A$.

And then, use the controller gain $K=[1\quad 0.9]$, and let $c=0.755$ (they are likewise from \cite{Li2011}).
Finally, the inner coupling matrix for this network can be constructed as
$$
H_6 = -\begin{pmatrix}
   1  & 0.9 \\
   -1  &  -0.9
\end{pmatrix}.
$$
The corresponding synchronization response is shown in Fig. \ref{subfig:6a}. The synchronization time is about $145s$.
\begin{figure}[H]
  \centering
  \includegraphics[width=2.5in]{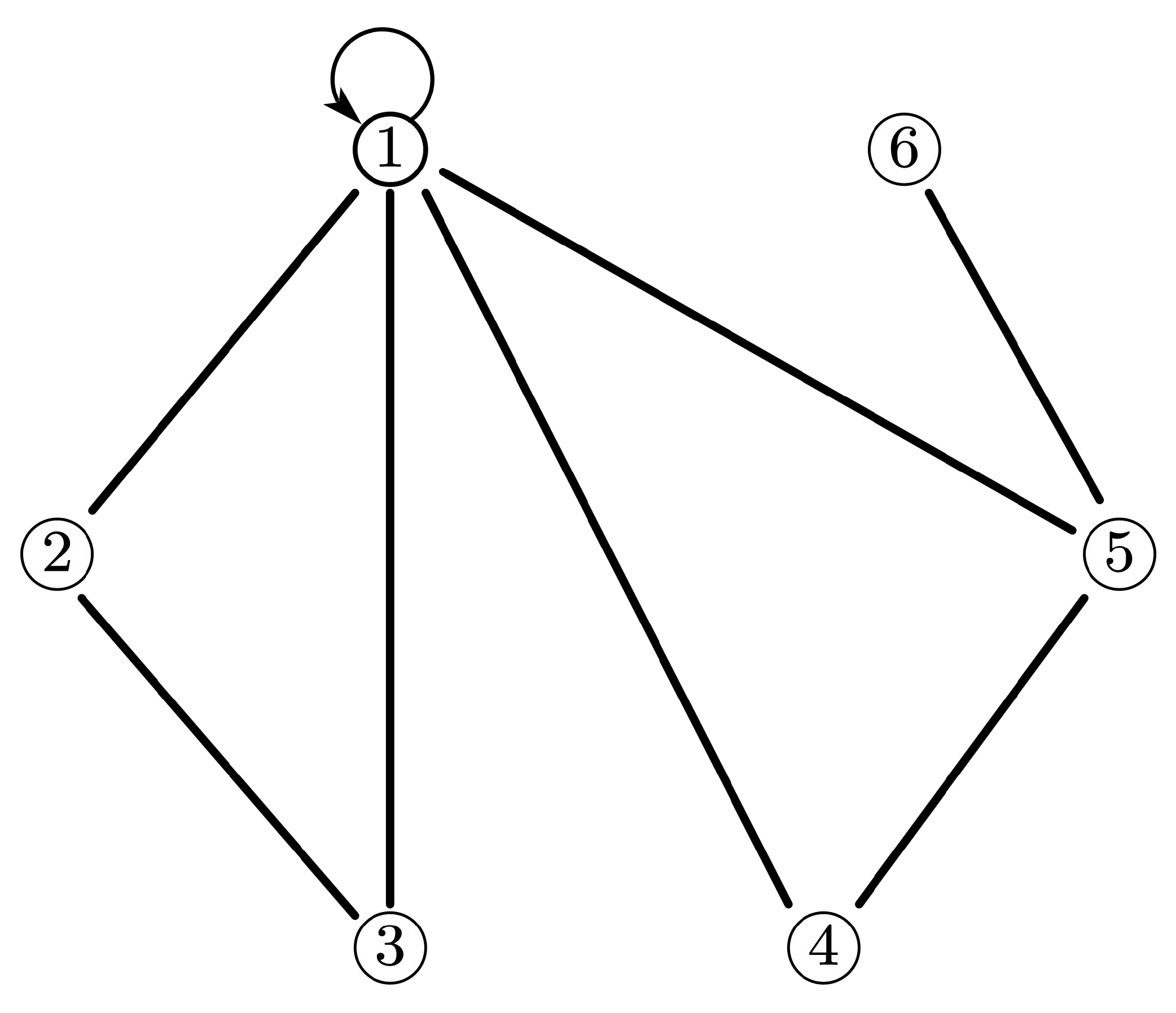}
  \centering \caption{The communication topology on Example 3 \cite{Li2011}}\label{f55}
\end{figure}
\begin{figure}[H]

    \begin{subfigure}{3.5in}
        \includegraphics[width=3.5in]{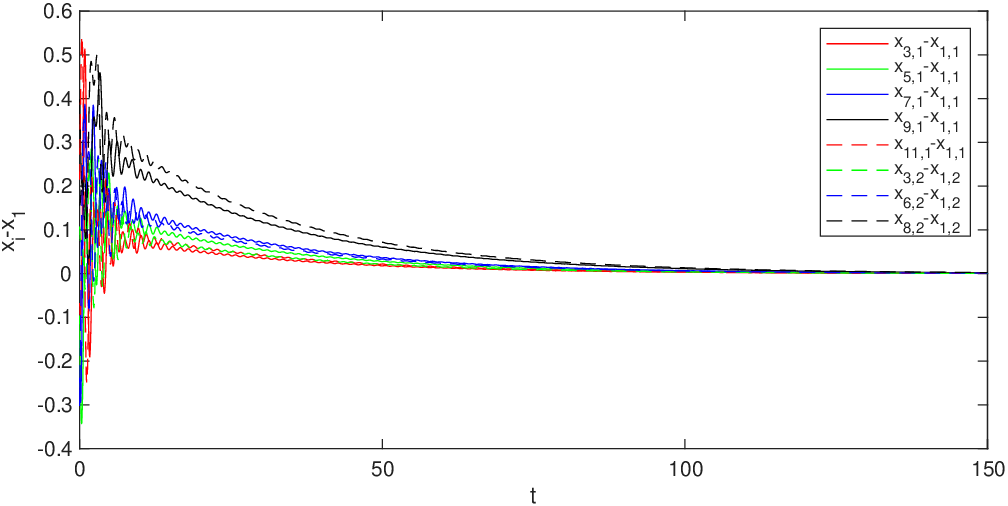}  
        \caption{Configuration result of the coupling matrix $H_6$}
        \label{subfig:6a}
    \end{subfigure}
    \vspace{0.5cm}  

    \begin{subfigure}{3.5in}
        \includegraphics[width=3.5in]{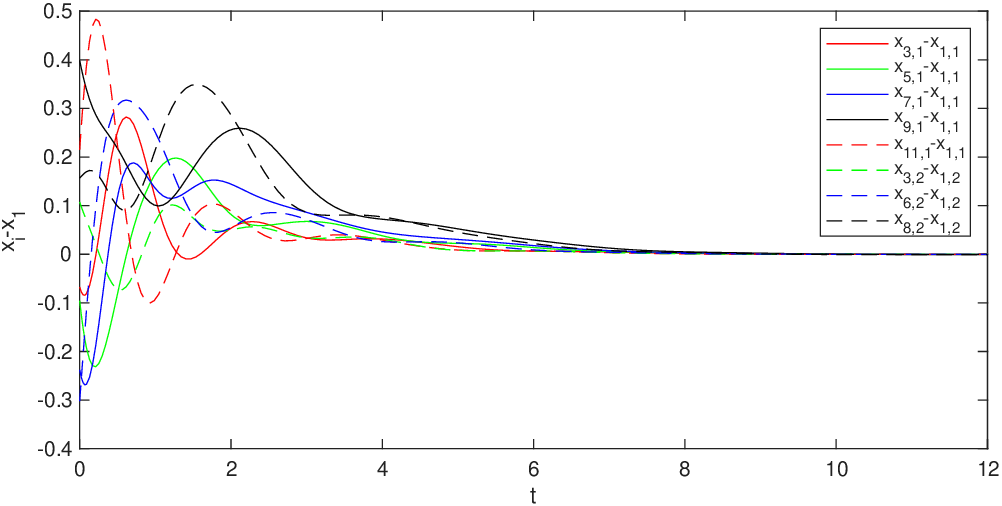}  
        \caption{Configuration result of the coupling matrix $H_7$}
        \label{subfig:6b}
    \end{subfigure}
\centering
    \caption{Configuration result of the coupling matrix $H$ for Fig. \ref{f55}}
    \label{f6}
\end{figure}

While we try to modify the input matrix as $B=[1\quad -2]^\prime$, the inner coupling matrix $H$ is updated, as follows,
$$
H_7 = -\begin{pmatrix}
   1  & 0.9 \\
   -2  &  -1.8
\end{pmatrix}.
$$
Under the effect of $H_7$, the performance of network synchronization obviously shows the shorter convergence time and smaller overshot in Fig. \ref{f6}.
It shows that the constructed input matrix $B$ affects the synchronization performance. Thus, how to construct it will be an open question.
Note that although this method can steer the network synchronization, it is hard to achieve at will pole placement, as in Section \ref{s1}. Undeniably, it is a drawback of this method. 

\subsection{Design controller gain $K$  by the inner coupling matrix of complex networks system}
In contrast, we can also design the controller gain $K$ for a multi-agent system using the inner coupling matrix $H$ of a complex networks system.
First, the inner coupling matrix $H$ is configured using the method in Section \ref{s1}.
Then, the controller gain $K$ of the multi-agent system can be calculated by the following formula from the transformation of Eq. (\ref{e14})
\begin{align} \label{e16}
K= - B^+ H,\quad   B^+=(B^TB)^{-1}B^T
\end{align}
Note that here the matrix $B$ must be full column rank and $(A,B)$ is controllable. There exists the following fact.
\begin{fact} \label{p2}
The consensus problem of the multi-agent system (\ref{e11}) and (\ref{e12}) can be solved by designing the controller gain $K$ (\ref{e16}), with the aid of the inner coupling matrix $H$ on the synchronization problem of the complex network (\ref{E2}),
if it satisfies three conditions,
 \begin{enumerate}
   \item $(A,B)$ is controllable;
   \item $B$ is full column rank;
   \item $B^+H \neq \textbf{0}$.
 \end{enumerate}
$K$ is the feedback gain of controller. $B$ is the input matrix of multi-agent system. $H$ is the configured coupling matrix of the complex networks system.
\end{fact}
\textbf{Example 4:} In this example, the dynamic and the topology are from {Example 3}.
The control input matrix of the multi-agent system is set to $B=[1\quad -1]^\prime$.
Obviously, $B$ is full column rank and $(A,B)$ is controllable, then the control gain $K$ can be calculated by Eq. (\ref{e16}).
For example, in light of $H_6$, the corresponding control gain is $K=[1\quad 0.9]$. In fact, it is the inverse operation of Example 3.
Therefore, the simulation results are the same as those in Fig. \ref{f6}. 

\section{Synchronization on nonlinear complex networks system}\label{s3}
In fact, our method  can still be used in synchronization on the nonlinear complex networks system.
In general, the complex network consists of some oscillators that are seen as $N$ nodes.
Let the isolated (uncoupled) dynamics be
\[
\dot{x}_i = F(x_i)
\]
for each node.
$x_i$ is the $n$-dimensional vector of dynamical variables of the $i_{th}$ node.
Here, $F(\cdot): \mathbb{R}^n \to \mathbb{R}^n$ is an arbitrary function of the variables of each node. Thus, the dynamics of the $i_{th}$ node over the complex network (referring to \cite{Pecora1998}) is written as Eq. (\ref{ee1}). Next, we could start with a more general case, nonlinear form in the inner coupling function, i.e. the system (\ref{ee1}).
Then, the collective system is Eq. (\ref{ee2}).

Then, linearize the system (\ref{ee2}) by the Jacobian matrix and
let $\xi_k$ be the variations on the $i_{th}$ node, and then
 the variational equation of Eq. (\ref{ee2}) is
\begin{align}\label{e18}
\dot{\xi}_k = \left[ I_N \otimes D \mathbf{F} + \sigma \mathcal{L} \otimes D \mathbf{H} \right] \xi_k
\end{align}
where $D\mathbf{F}$ and $D\mathbf{H}$ are Jacobian functions of $\mathbf{F}(\cdot)$ and $\mathbf{H}(\cdot)$ respectively.
Diagonalizing the matrix $\mathcal{L}$, the system is recast as Eq. (\ref{e3})
\begin{align}\label{e3}
\dot{\xi}_k = [D\mathbf{F} + \sigma \gamma_k D\mathbf{H}]{\xi}_k,\quad k=1,2,\cdots,N. 
\end{align}

We next divide the matrices $D\mathbf{F}$ and $D\mathbf{H}$ into two parts, respectively, that is, 
$D\mathbf{F}=\Phi_1+\Phi_2(\mathbf{x})$ and $D\mathbf{H}=\Psi_1+\Psi_2(\mathbf{x})$, where $\Phi_1$ and $\Psi_1$ are the constant matrix, $\Phi_2(\mathbf{x})$ and $\Psi_2(\mathbf{x})$ are function matrices with respect to $\mathbf{x}$. Therefore, Eq. (\ref{e3}) is recast as  
\begin{align}\label{e4}
\dot{\xi}_k=[\Phi_1+\sigma \gamma_k \Psi_1] {\xi}_{k}+[\Phi_2(\mathbf{x})+\sigma \gamma_k \Psi_2(\mathbf{x})] {\xi}_{k}.
\end{align}
The stable problem on (\ref{e3}) is converted into one of whether the first and second items are Hurwitz. Firstly, in light of the linear transformation  ${\xi}_k=P\tilde{\xi}_k$, the first item in (\ref{e4}) is written as  
\begin{align}\label{e5}
\dot{\tilde{\xi}}^1_k=&(J+ \sigma \gamma_k \mathcal{H}) \tilde{\xi}^1_{k},\quad k=2,3,\ldots,N.
\end{align}
Therefore, the first item can be guaranteed to be Hurwitz using the method in Section \ref{s1}.

The above method ensures that the first item in (\ref{e4}) is Hurwitz. 
Now, if the second item is stable (that is, $(\Phi_2(\mathbf{x}) + \sigma \gamma_i\Psi_2\mathbf{x}))$ is negative definite) with varying of state $\mathbf{x}$ at every time or vanishes, the system (\ref{e3}) will achieve synchronization.

Next, we focus on the second item. Due to that the stability just depends on the system's real part, and then to ensure that holds for all $i$, we consider it as follows:
\begin{align}\label{e7}
\dot{\tilde{\xi}}_k^2 = \left(\Phi_2(\mathbf{x}) + \sigma Re({\gamma}_k)\Psi_2(\mathbf{x})\right) \tilde{\xi}_k^2.
\end{align}
Setting $\Psi_2(\mathbf{x})=-\frac{1}{\kappa}\Phi_2(\rho_j \mathbf{x})$, with $\kappa \leq \min\limits_{\gamma_i}{Re(\gamma_i)}$. While $\rho_j$ are the configuration parameters of the nonzero elements, which are used to guaranty all $Ger\check{s}gorin$ discs in the matrix $[D\mathbf{F} + \sigma \gamma_k D\mathbf{H}]$ are not in the right half-plane. Inserting $\Psi_2(\mathbf{x})$ into (\ref{e7}) without considering the imaginary part, one has 
\begin{align}\label{e8}
\dot{\tilde{\xi}}_k^2 = \left( \Phi_2(\mathbf{x}) - \frac{\sigma Re({\gamma}_k)} {\kappa} \Phi_2(\rho_j \mathbf{x})\right) \tilde{\xi}_k^2. 
\end{align}
Thus, if only $\Psi_2(\mathbf{x})$ is designed to guaranty that the system (\ref{e8}) is stable or eliminate $\Phi_2(\mathbf{x})$, the complex network achieves synchronization. 

Note that the configuration way for $\Psi_2(\mathbf{x})$ is feasible, since the local states $\mathbf{x}$ are measurable or observable in general.
\hspace{20in} 

\textbf{Example} --- An example of the three-oscillator universal probe \cite{Fink2000} shows how this works.
The equations of motion for the $R\ddot{o}ssler$ system are
\begin{align}\label{e9}
\frac{dx_i}{dt} &= -(y_i+z_i), \nonumber \\
\frac{dy_i}{dt} &= x_i + ay_i, \\
\frac{dz_i}{dt} &= b + z_i(x_i - c),\quad i = 1,2,3 \nonumber
\end{align}
where $a=b=0.2$ and $c=7.0$. It has a chaotic regime of behavior.

Consider the following network configuration of three coupled oscillators
(for now, only additive coupling is considered):
\begin{align}\label{e10}
\frac{d\mathbf{x}_i}{dt} &= F(\mathbf{x}_i) + \frac{\varepsilon}{3}\bigl[{H}(\mathbf{x}_{i+1}) + {H}(\mathbf{x}_{i-1}) - 2{H}(\mathbf{x}_i)\bigr] \nonumber \\
&\quad + \frac{\delta}{\sqrt{3}}\bigl[{H}(\mathbf{x}_{i+1}) - {H}(\mathbf{x}_{i-1})\bigr],\ i = 1,2,3 \quad \mbox{cyclically}, 
\end{align}
where $\mathbf{x}_i=[x_i\quad y_i\quad z_i]^\prime$ is the state of the node. The corresponding network schematic is shown in Fig. \ref{f7}, 
where each oscillator
is connected to the two neighbors through the $H(\cdot)$ function with combinations that are symmetric ($+$, with
weight $\varepsilon/3$) and antisymmetric ($-$, with weight $\delta \sqrt{3}$). If we let  $\varepsilon=3$ and $\delta=\sqrt{3}$, the states of three oscillators can clearly achieve synchronization with the given inner coupling matrix $H=\begin{pmatrix}
0 & 0 & 0 \\
0 & 1  & 0 \\
0 & 0 & 0
\end{pmatrix}$. It has already been verified in \cite{Fink2000}. 

Unfortunately, the three-oscillator universal probe is impossible to synchronize when $\varepsilon$ is $0.1$ or smaller. At this point, there is nothing to do to update the network structure in reality under the constraints ($\varepsilon \le 0.1$). However, the method in this paper can be used to address it by restructuring the inner coupling matrix. which can be designed for the node locally.
 \begin{figure}[http]
  \centering
  \includegraphics[width=3in]{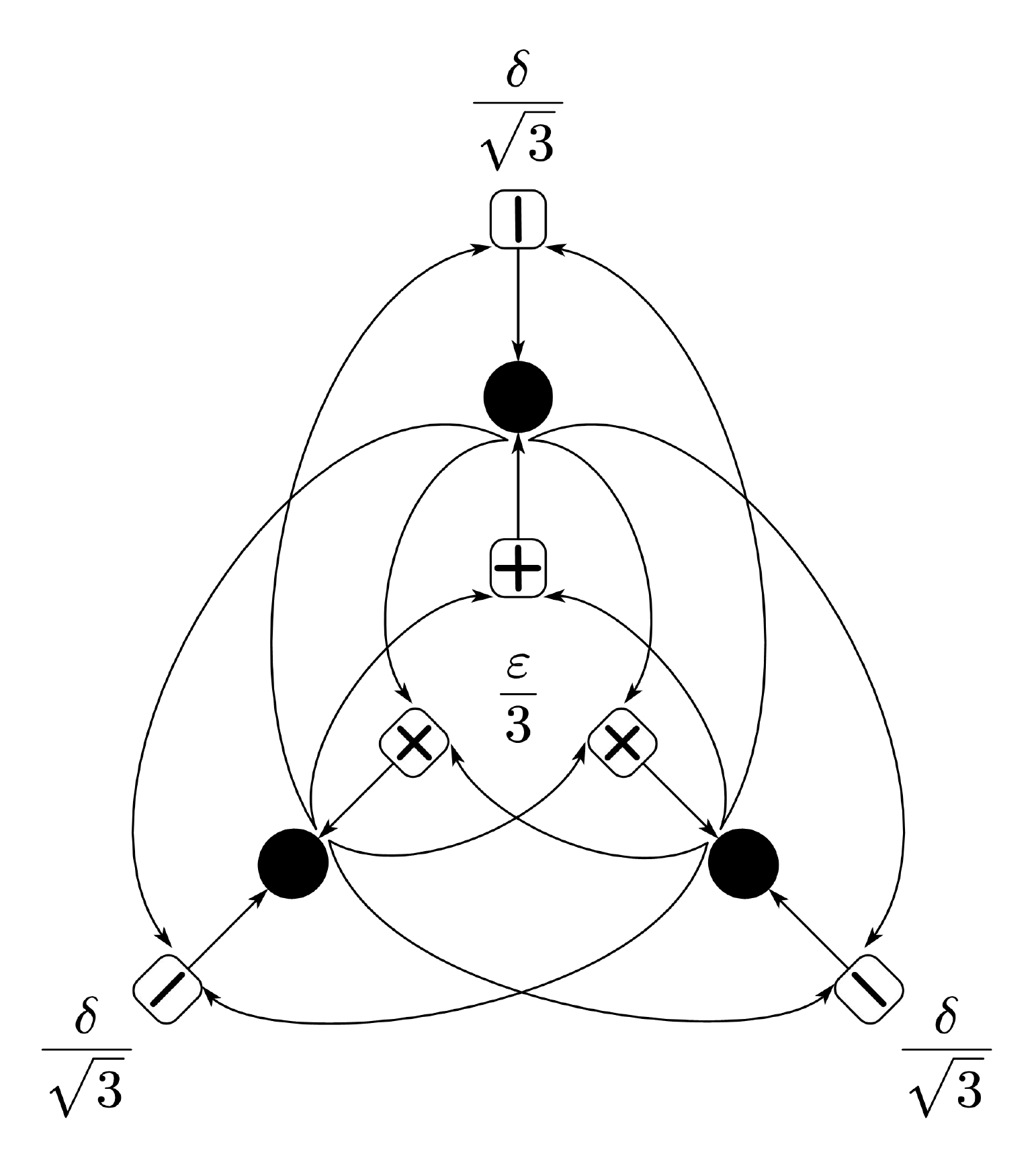}
  \centering \caption{Network schematic of three oscillator system }\label{f7}
\end{figure}

First, the collective variational equation of Eq.(\ref{e3}) is recast as
\begin{align}\label{e11}
\frac{d\xi}{dt} &=(I_N\otimes D\mathbf{F}) \xi +   (G \otimes D\mathbf{H})  \xi, 
\end{align}
where $\xi=[\xi_1,\ldots,\xi_N]$, without loss of generality, let $\sigma=1$. And then
$$G=
\left(
\begin{array}{ccc}
-2\frac{\varepsilon}{3} & \frac{\varepsilon}{3} + \frac{\delta}{\sqrt{3}} & \frac{\varepsilon}{3} - \frac{\delta}{\sqrt{3}} \\
\frac{\varepsilon}{3} - \frac{\delta}{\sqrt{3}} & -2\frac{\varepsilon}{3} & \frac{\varepsilon}{3} + \frac{\delta}{\sqrt{3}} \\
\frac{\varepsilon}{3} + \frac{\delta}{\sqrt{3}} & \frac{\varepsilon}{3} - \frac{\delta}{\sqrt{3}} & -2\frac{\varepsilon}{3}
\end{array}
\right)
$$

Diagonalizing the connection matrix $G$ in the second term of Eq.(\ref{e11}), one has
\begin{align}\label{e12}
\dot{\xi}&=
(I_N\otimes D\mathbf{F}) \xi
 +
\begin{pmatrix}
0 & 0 & 0 \\
0 & \gamma_1 D\mathbf{H} & 0 \\
0 & 0 & \gamma_2 D\mathbf{H}
\end{pmatrix} \xi,
\end{align}
where $\gamma_{1,2}=\varepsilon \pm i\delta$ and 
$
D\mathbf{F}=\begin{pmatrix}
0 & -1 & -1 \\
1 & a & 0 \\
z & 0 & x-c
\end{pmatrix}.
$

Here let
$
D\mathbf{F}=\Phi_1+\Phi_2(\mathbf{x}),
$
$$
\Phi_1=\begin{pmatrix}
0 & -1 & -1 \\
1 & a & 0 \\
0 & 0 & -c
\end{pmatrix},\quad
\Phi_2(\mathbf{x})=\begin{pmatrix}
0 & 0 & 0 \\
0 & 0 & 0 \\
z & 0 & x
\end{pmatrix}.
$$
\begin{figure}[http]
        \includegraphics[width=3.5in]{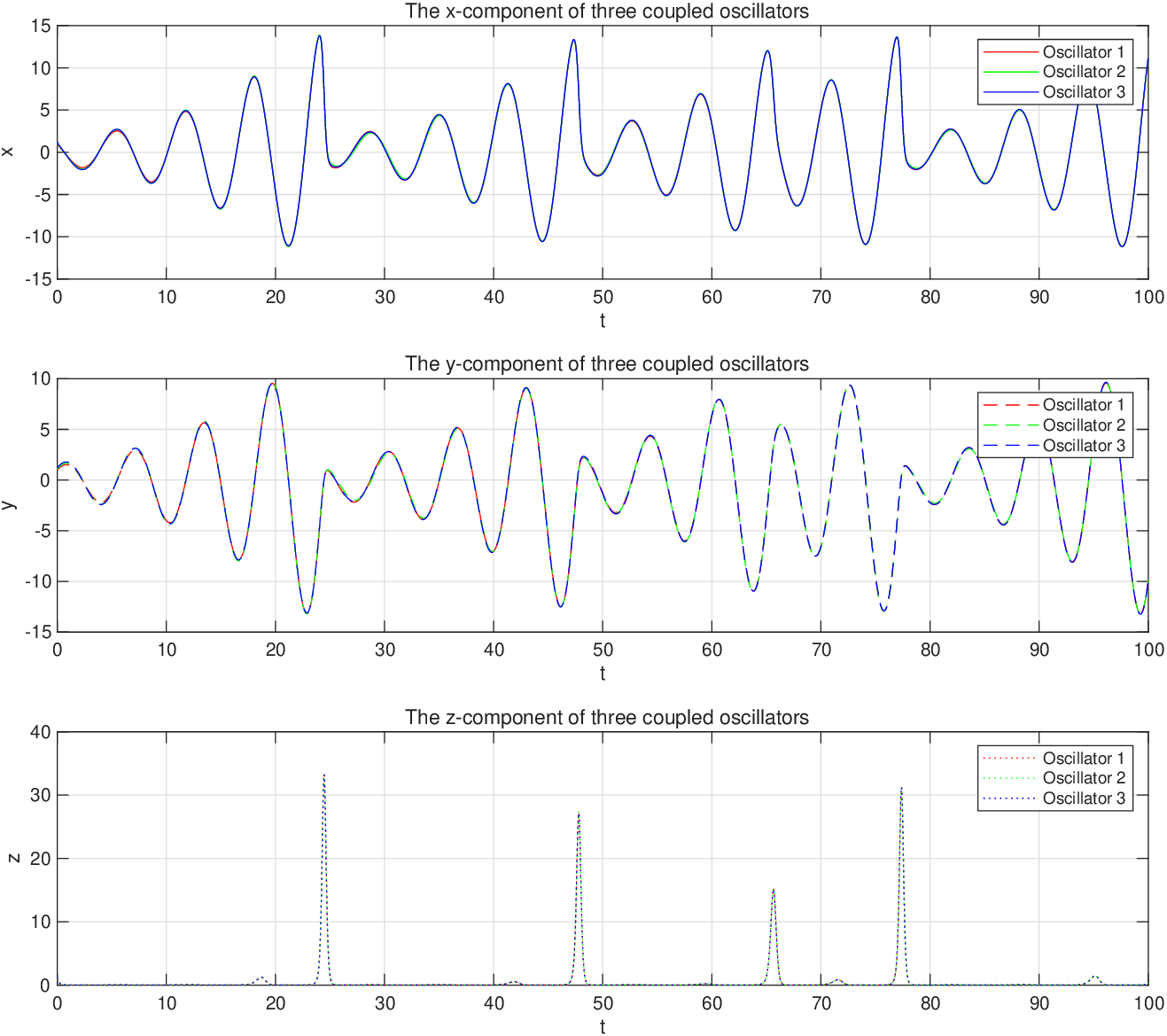}  
        \caption{The component evolution of three coupled oscillators}
        \label{f3}
\end{figure}
Set $\varOmega (\mathbf{x})=-\Phi_2(\mathbf{x})$,
and let $\Psi_2(\mathbf{x})=\frac{1}{\varepsilon}\varOmega(\xi)$, the second item in (\ref{e4}) vanishes.
Moreover, set $\mathcal{H}=\mbox{diag}\{-5\}$, and then $\Psi_1=\mbox{diag}\{-5\}$ by $Re(P\mathcal{H}P^{-1})$. One has
$$D\mathbf{H}= \Psi_1+\Psi_2(\mathbf{x})=\begin{pmatrix}
-5 & 0 & 0 \\
0 & -5 & 0 \\
-\frac{1}{\varepsilon}z & 0 & -5-\frac{1}{\varepsilon}x
\end{pmatrix},$$
where $\varepsilon=0.1$ and $\delta=\sqrt{3}$.
Furthermore, the inner coupling function $H(\mathbf{x})= D \mathbf{H}$.

The simulation results are shown in Figs. \ref{f3} and \ref{f4}.
Fig. \ref{f3} describes the evolutionary process of three components $(x,y,z)$ respectively.
We can see that three oscillators do not achieve synchronization nearly until $40s$, more clearly manifested in the $x$ and $y$ components.
Fig. \ref{f4} shows the result $3-D$ in the $x-y-z$ coordinate system.
Clearly, the three oscillators spontaneously achieve synchronization under the influence of the inner coupling matrix. 
\begin{figure}[http]
        \includegraphics[width=3.5in]{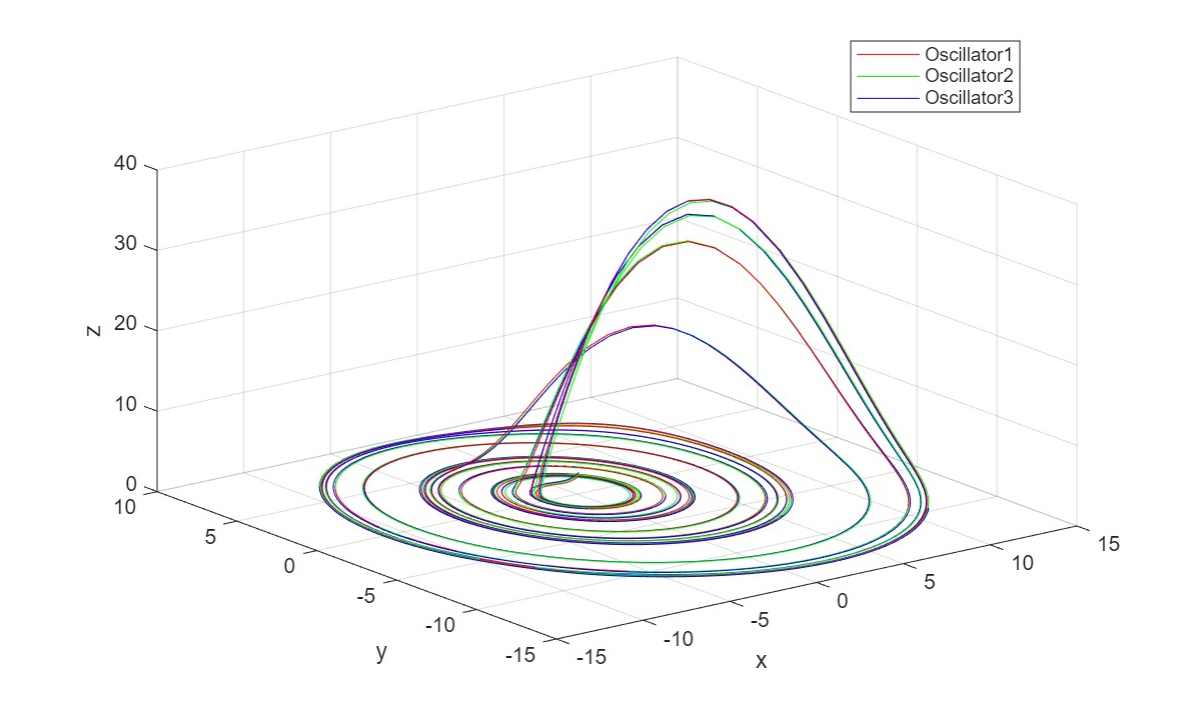}  
        \caption{The $3-D$ state evolution of three coupled oscillators}
    \label{f4}
\end{figure}

{\color{red}\section{Conclusion} This paper gives a fresh look at network
synchronization. Motivation comes from the fact that the inner
coupling matrix is not subject to any restriction, such as distance
and communication strength among nodes. It can be configured
at will to meet the synchronization performance if only the
states of the local dynamic are measurable or observable and
the communication topology is connected. Thus, it is very useful
for future practical engineering design.}

\bibliographystyle{IEEEtran}
\bibliography{mydatabase}
\end{document}